\documentclass{sig-alternate-per}

\usepackage{amsmath}

\usepackage{cite}

\usepackage{url}

\usepackage{color}

\newtheorem{theorem}{Theorem}

\newtheorem{definition}{Definition}

\newcommand{\be}{\begin{equation}}
\newcommand{\ee}{\end{equation}}

\definecolor{co}{rgb}{0.8,0,0.8}
\definecolor{gr}{gray}{0.5}

\catcode`\@=11

\catcode`\@=12

\usepackage{comment}

\begin{document}

 \conferenceinfo{Symposium on Cryptocurrency Analysis (SOCCA) 2020}{~~~Milan, Italy}

\title{Blockchain Privacy Through Merge Avoidance and Mixing Services: a Hardness and an Impossibility Result}


\author{Jefferson E. Simoes\\
 \affaddr{UNIRIO} \\
\affaddr{Rio de Janeiro, Brazil} \\ 
 \and
Eduardo Ferreira \\
\affaddr{UFRJ } \\
\affaddr{Rio de Janeiro, Brazil} \\ 
\and
Daniel S. Menasche \\
\affaddr{UFRJ} \\
\affaddr{Rio de Janeiro, Brazil} \\
\and
Carlos A. V. Campos \\
\affaddr{UNIRIO} \\
\affaddr{Rio de Janeiro, Brazil} 
}


\maketitle

\begin{abstract}
Cryptocurrencies typically aim at preserving the privacy of their users. Different cryptocurrencies preserve privacy at various levels, some of them requiring  users to rely on strategies to raise the privacy level to their needs. Among those strategies,  we focus on two of them: merge avoidance and mixing services.  Such strategies may be adopted on top of virtually any blockchain-based cryptocurrency. In this paper, we show that whereas optimal merge avoidance leads to an NP-hard optimization problem, incentive-compatible mixing services are subject to a certain class of impossibility results. Together, our results contribute to the body of work on fundamental limits of privacy mechanisms in blockchain-based cryptocurrencies.
\end{abstract}






\section{Introduction}
 
 Privacy is one of the desired properties  by most  users issuing online transactions. Privacy concerns have motivated the development of novel cryptocurrencies, and rankings of cryptocurrencies with respect to privacy.  Monero and Zcash are highly ranked in that respect,  whereas Bitcoin and Ethereum are known to suffer from   weaknesses~\cite{tang2020privacy, monero}.
 

Although  
the history of transactions in blockchain systems is publicly  available to anyone through the blockchain itself, the owners of the addresses involved in the transactions in the blockchain are in principle unknown.  
However, a number of  techniques have been devised and employed to de-anonymize  addresses. They range from basic heuristics, such as clustering the input addresses of transactions~\cite{spagnuolo2014bitiodine}, to some more advanced techniques such as applying machine learning to infer the owners of addresses~\cite{ermilov2017automatic}. 
Merge avoidance and mixing services are two of the solutions to mitigate the effectiveness of such de-anonymization  techniques~\cite{narayanan2017obfuscation,  sarfraz2019privacy,  wu2020detecting}.



Users in Bitcoin store their coins in wallets, and exchange funds through transactions, with each transaction comprising its inputs (funds sources) and outputs (funds recipients). A privacy issue consists in exploring the regularity of certain types of transactions to infer personal information.

  If Alice, an employee of a coffee shop, is paid her salary in Bitcoin, one can infer Alice's salary by noticing the existence of weekly or monthly transactions with the coffee shop's wallet as input and Alice's wallet as output. 
To avoid such an issue, one strategy consists in Alice creating multiple wallets  to receive her salary.

 In the above example, the situation involved  recurring (separated in time) transactions to/from the same address.   
 The lack of privacy is best dealt with by avoiding address reuse (including avoiding reuse for change transactions). 
 Alternatively, a  single merged transactions (occuring at one point in time), may also involve multiple input addresses.
 %
 %
 %
 %
 %
 %
 Returning to Alice's example, it may be that the company pays all of its employees at the same time through a single merged transaction. Then if one knows Alice's wallet one can readily identify the wallets of the other employees of the company and how much they are paid using clustering techniques~\cite{meiklejohn2013fistful}.
In addition, if one knows that one of the addresses feeding the merged transaction is owned by the coffee shop, the coffee shop ownership of the other addresses can be inferred. 
 In that case, one may want to avoid the merge of transactions into a single sink. Such a strategy is referred to as \emph{merge avoidance}~\cite{mike_hearn}.

\begin{table}[t]
    \centering
    \scalebox{0.85}{
    \begin{tabular}{l|l|l}
         & merge avoidance & mixing
         service  \\
         \hline
         \hline
        goal & \multicolumn{2}{c}{ obfuscate the network} \\
        \hline
        strategy & \multicolumn{2}{c}{edge insertion, i.e., add transactions} \\
        \hline
        relies on & sender and receiver &  network of users (relays) \\
        \hline
        our  & optimal solution is  &   incentive compatibility   \\
    contribution   &  NP hard & jointly for sender and   \\
    & &  relays is  impossible \\
        \hline
    \end{tabular} } \vspace{-0.1in}
    \caption{Privacy schemes and our contributions} \vspace{-0.2in}
    \label{tab:table1}
\end{table}

An alternative strategy adopted by users to increase their privacy consists of relying on \emph{mixing services}.  Mixing services comprise users that produce spurious transactions that receive funds from given inputs, split those funds across multiple wallets, possibly multiple times, before  finally transferring those funds to the intended targets.  The purpose of mixing services is to obscure the trail of transactions back to the fund's original source, 
mixing potentially identifiable or ``tainted'' cryptocurrency funds with others.

It follows from the above discussion  that there are two main types of linkage in a financial system we want to
protect against: 
(1) linking \textbf{different} coins/transactions from the \textbf{same user} together, e.g., to avoid the possibility of determining that two coins are owned by the same user,
(2) linking the \textbf{same} coin over (possibly) \textbf{different} users together, e.g., to avoid the possibility of determining that a given coin was transferred between two given  users.
Primary techniques for solving (1) and (2) are merge avoidance and  mixing, respectively. This paper shows that both are
difficult to do robustly, through our   main results: we show that while optimal merge avoidance leads to an NP-hard optimization problem, mixing services yield   impossibility results pertaining to their incentive mechanisms (Table~\ref{tab:table1}). 



\textbf{Prior art. } There has been a recent surge in interest on the privacy aspects of cryptocurrencies~\cite{tang2020privacy, tikhomirov2020quantitative}.
Merge avoidance was first proposed as an alternative to services such as CoinJoin~\cite{maurer2017anonymous}.\footnote{\url{https://medium.com/p/7f95a386692f/}}
Among the tools that leverage merges to break anonymity, Bitiodine is a notable example~\cite{spagnuolo2014bitiodine}.
Mixing services have their roots in networks such as the Tor project~\cite{tor}, which aims to anonymize users in the Internet.  The literature on mixing services has been rapidly growing~\cite{moser2013inquiry, ziegeldorf2015coinparty}, with measurements and results that are complementary to those reported in this work. 
   
\textbf{Contributions. } Our contributions are twofold. 

\textit{\textbf{Optimal merge avoidance is NP-hard: } } we show that optimally allocating resources for merge avoidance is NP-hard.  The result involves a reduction from the partition problem, and motivates  heuristics for merge avoidance. 



\textit{\textbf{An impossibility result on incentive-compatible mixing services: }} we prove  that no incentive-compatible strategy for resource allocation in mixing services can at the same time account for strategic relays and strategic targets.   Resource allocation for incentive-compatible mixing services must  relax   assumptions, e.g., pertaining to constant amount of currency  in the system, to be feasible.



%



\section{Merge Avoidance}
\label{sec:hard}


Merge avoidance in blockchain systems was first proposed by Mike Hearn~\cite{mike_hearn}, one of the Bitcoin core developers, by the end of 2013.  It basically consists of splitting transactions among multiple addresses, so as to avoid the easy identification that two addresses are owned by the same user.  Our goal is to measure the impact and effectiveness of merge avoidance, noting that the split of transactions involves costs.


We  model the Bitcoin transaction network as a graph comprising nodes of two types:

 \textbf{Value node:} 
    each transaction has a list of input value nodes and generates a list of output value nodes.
     Each input and  output must correspond to an address that holds some amount of currency (e.g., Bitcoins or BTC).

  \textbf{Transaction node:} this node characterizes the action of transferring value from input to output nodes. Each transaction corresponds to a transaction identifier in the blockchain.

Directed edges are created between a transaction node and each value node involved in this transaction, with directions from input value nodes to the transaction node and from the transaction node to the output nodes; see Figure~\ref{alg:original_graph}(a).

The BTC graph  spawns from its genesis block.  Every new block generates at least one new value node, corresponding to the coinbase transaction, with a given reward. Value nodes are set as input to  transaction nodes, which output one or more new value nodes. It is well  known that roughly 80\% of the transactions in the BTC network  output 2 value nodes, one corresponding to the destination and the other corresponding to a change to the  transaction issuer.   

One of the most natural strategies to cluster  addresses is to merge the input addresses for a transaction and assume that they all belong to the same owner. That is because it is likely that the owner of some set of addresses will use the accumulated values together to pay for some goods or service, and output the remainder to one change address. Now, suppose that one  learns the owner of one of those addresses, e.g., through side channels, like online forums or social networking, by social engineering, or by leveraging the fact that the address belongs to  a famous retailer, who issues transactions at known values. Then,  one can easily learn the owner of additional addresses belonging to that user,  following the activity of his addresses and, subsequently, also de-anonymize other elements in the network.

The goal of merge avoidance~\cite{mike_hearn}  is to mitigate this weakness by avoiding the use of multiple input addresses into a transaction node, effectively modifying the transaction graph. To that aim, if one intends to merge  $N$ addresses to pay a single destination address, one alternative  is to create $N$ separate transactions and output them to novel destination addresses belonging to the same user. The main cost involved  in this scheme relates to the  extra fees incurred in the issuing of such transactions. 


\subsection*{Hardness of Optimal Merge Avoidance}


We begin by considering a single transaction, illustrated in Figure~\ref{alg:original_graph}. The circles in Figure~\ref{alg:original_graph} represent value nodes and the boxes represent transaction nodes, and the edge values are the amount of BTC transferred. Figure~\ref{alg:original_graph}(a) shows the original graph, before transaction splitting for merge avoidance, and Figure~\ref{alg:original_graph}(b) shows one possible outcome of merge avoidance --- the owner of the addresses $D1$ and $D2$ is the same user.
To simplify presentation, we are  not considering any fees in this simplified view of the problem. Then, the sum of the values from the input  equals the sum of values in the output  and we work under such assumption. 
%
%

A user builds a transaction, and before submitting it to the blockchain the transaction goes through  merge avoidance.  The output of merge avoidance is a new set of transactions, whose outputs  fulfill  the originally requested output values. The user then submits those transactions to the miners, to  add them into the blockchain.




\begin{definition}[Multi-target merge avoidance] 
An original  transaction has $\ell$ input and $r$ output value nodes, with integer values $s_i$ and $t_j$, $1 \leq i \leq \ell$ and $1 \leq j \leq r$, respectively. 
A set of modified  values comprises integer elements ${m}_{i,j}$ such that
\begin{align}
    s_i &= \sum_{j=1}^r m_{i,j}, \quad  t_j = \sum_{i=1}^\ell m_{i,j}   \label{eq:constr}
\end{align}
While solving the merge avoidance problem, the number of transactions in the modified graph corresponds to the number of strictly positive values in the set of modified values.   The aim is to minimize such number, subject to~\eqref{eq:constr}.
\end{definition}
Note that $m_{i,j}$ is the amount of coins routed from $s_i$ to $t_j$ under merge avoidance.    Note also that we make no assumptions about the semantics of the output values, e.g., regarding what is an effective transfer and what is a change. 

The complexity of finding the minimum set of merge avoidance transactions  is established by the following result.


\begin{theorem}
Multi-target   merge avoidance is NP-Hard.
\end{theorem}

Let $\ell$ and $r$ be the number of value nodes in the left and right side of the transaction.
A lower bound on the number of transactions  in the modified graph is $\ell$. An upper bound is given by $\ell \cdot r$, which occurs when we need to transfer value from each of the   nodes in the left to each of the   nodes in right. In that case, the original graph is transformed into a complete bipartite graph with  $\ell \cdot r$ edges, which in turn   yields a modified graph with  $\ell \cdot r$ transactions.


\begin{figure}[t]
\centering 
\includegraphics[width=0.3\textwidth]{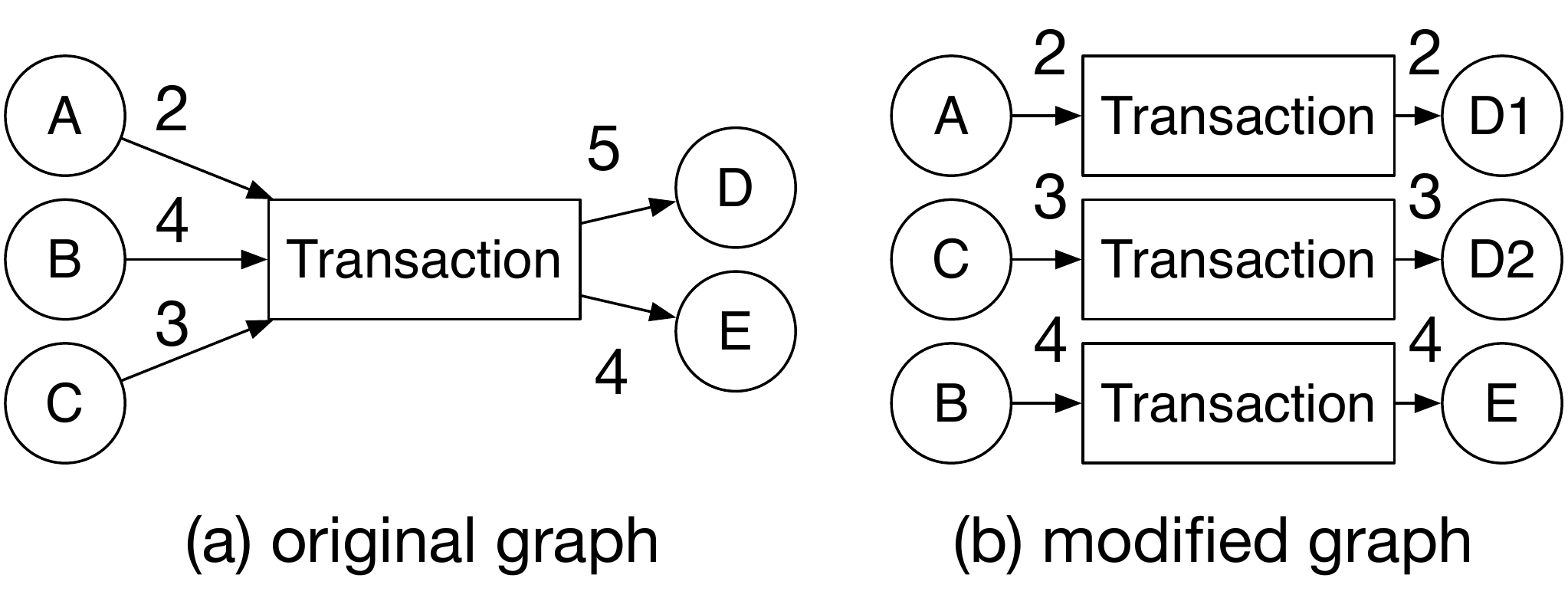}   \vspace{-0.15in}
\caption{Illustrative transaction graphs.}  \vspace{-0.1in}
\label{alg:original_graph}
\end{figure}

\begin{proof}[of Theorem 2.1]
The proof follows from a reduction from the partition problem, which is NP-complete, to a special instance of merge avoidance wherein we have $N$ input nodes and 2 output nodes, whose values   equal   half the sum of the values of the input nodes.
\end{proof}

There is a pseudo-polynomial time dynamic programming solution to the partition problem, and  heuristics to solve the problem in many special instances. Nonetheless, new heuristics may still be required to account for the multiple layers of transactions in the BTC graph.

Note that the above discussion accounted for a single transaction. We leave the analysis of heuristics for the  multi-transaction setup as subject for future work, noting that they may benefit from  approximate solutions to the single transaction instance.

\textbf{Change transactions. } In the above definition there is no special  treatment of change transactions.  Nonetheless, 
 typical transactions in Bitcoin have a single output target value and a change.
As those transactions are prevalent, we  consider the  following   formulation of the problem.
\begin{definition}[Single-target merge avoidance]
Given a set of $\ell$  value nodes with positive integer values $s_i$ and one target output value $v$,   we want to determine the set $K$, contained in the power set of  ${\{1,\ldots, \ell\}}$,  with smallest cardinality, that meets the  output,
\begin{align}
\min_{K \in P({\{1,\ldots, \ell\}})} |K| \textrm{ such that }     \sum_{i \in K} s_i \geq v.
\end{align}
The smallest set corresponds to the optimal way of creating a merge avoidance set of transactions,  minimizing the number of transactions while fulfilling the requested output value.
\end{definition}

Under the above problem formulation,   merge avoidance can be solved in polynomial time using a  greedy algorithm. The algorithm first orders the input nodes of a transaction,  in decreasing order based on their values. Then, it selects the first nodes whose sum of values is sufficient to resolve the considered transaction. 




\section{Mixing Services}

\label{sec:mixing}



Another strategy to increase privacy in a blockchain system relies on the deployment of mixing services.
These services attempt to conceal the identity of nodes in blockchain records by ``routing'' these records through a sequence of proxy nodes, successively replacing the identity of each node with that of its successor.
This can be thought of as a ``shuffling'' of node identities, which does not remove the identity of the original node but rather hides it in a steganographic fashion.
While this idea is not novel in itself, tracing back its origins to Web anonymity systems such as Tor~\cite{tor},
it has become particularly appealing for cryptocurrency systems, as their records regard objects (quantities of some cryptocurrency) with direct monetary value.
Mixing services usually operate under a rewarding scheme in which the intermediate nodes involved in the concealment process charge the node interested in having its own identity concealed; these transactions are usually managed externally through e.g. a system of credits, though cryptocurrency systems can establish that these payments are performed within the system itself, which simplifies their implementation.

In designing such mixing services, it is important not to assume that every node in the blockchain system is trustworthy, and therefore to design the service such that it is robust to attacks from its own nodes.
In particular, we consider here \emph{edge insertion attacks}. In such attacks, nodes falsely claim that additional nodes were involved in the process in order to receive the corresponding rewards on their behalf.
This kind of attack is rather simple to prevent through the deployment of a centralized authority responsible for validating the service performed and verifying node identities, but becomes a challenge if both these tasks are designed to be done in a distributed fashion.

Next, we will work with the following definition for an edge insertion attack and associated terminology.




\begin{definition}
In an identity concealment process, an \emph{applicant node} is a node which requests that its identity in a blockchain record be concealed,
and a \emph{concealer node} is any of the nodes involved in concealing the identity of the applicant node.
A \emph{concealment route} is a sequence of nodes $\mathcal{R} := (r_0, r_1, \dots, r_{i-1}, r_{i}, \dots, r_n)$ with $r_0$ being the applicant node and all others being concealer nodes.
The \emph{length} of this route is given by the number $n$ of concealer nodes.
The \emph{reward} received by a concealer node $r_i$ is denoted by $R_i$, and the total \emph{cost} paid by the applicant node is denoted by $C$.
\end{definition}

\begin{definition}
An \emph{edge insertion attack}, performed by an attacking node $r_i$ $(0 \leq i \leq n)$, occurs when $r_i$ forges a set of \emph{Sybil nodes} $\mathcal{S}$, replaces $\mathcal{R}$ by a bogus concealment route $\mathcal{R}' = (r_0, \dots, r_{i-1}, s_1, \dots, s_j, r_{i}, \dots, r_n) := (r'_0, \dots, r'_{n+j})$, with $j \in\mathbb{N}$ and $s_1,\dots,s_j \in \mathcal{S}$, and acquires rewards directed to itself and all nodes in $\mathcal{S}$.
\end{definition}

This definition of concealment route only covers mixing by routing funds through a long sequence of nodes. Mixing services usually apply additional strategies involving e.g. splitting and re-merging funds; nevertheless, any impossibility results that apply to concealment routes extend to general mixing services, whose tools include concealment routes.

\subsection*{Impossibility Result}

In designing a distributed system which prevents edge insertion attacks, one can  leverage the freedom in determining the values of the rewards to be paid by the applicant node. One possibility is to diminish the value of the individual reward for longer concealment routes. This way, even if malicious nodes are able to acquire more reward quotas, their smaller value would not be enough to make up for the reward they were originally entitled to.

However, it is important to note that not only concealing nodes can attempt to perform edge insertion attacks. Rather, the applicant node can also engage in such attempts. In this case, however, the attacker is not attempting to hoard more reward quotas than he was entitled to, but to obtain a certain amount of rewards in order to inadequately recover a portion of the original cost.

The following theorem states that, under mild assumptions on the reward scheme, it is impossible to simultaneously prevent edge insertion attacks from concealment and applicant nodes. 


\begin{theorem}\label{nao-previne}
Consider a reward scheme for an identity concealment process. Assume it follows the guidelines below:
\begin{description}
\item[Zero-sum] The applicant node is charged for exactly the total reward paid to concealing nodes, i.e., the total amount of credits in the system is constant over time;\label{zero-sum}
\item[Length-dependency] The amount of credits the applicant node is charged for is exclusively a function of the length of the concealment route;
\item[Uniformity] Rewards paid to concealing nodes are uniform, i.e., all concealing nodes receive equal rewards.
\end{description}

Then, this scheme cannot prevent, simultaneously, edge insertion attacks performed by the applicant node and by concealment nodes.
\end{theorem}

\begin{proof}
Consider an identity concealment process with a concealment route $\mathcal{R}$ of length $l$.
Denote by $C(l)$ the total cost charged from the applicant node, and by $R(l)$ the reward received by a concealer node.
Since rewards are uniform, $R(l) = C(l)/l$ is equal for all concealer nodes. 

Now, denote by $CR(l,k)$ the net reward received by an attacking concealer node, when adding $k \geq 0$ Sybil nodes to $\mathcal{R}$.
If $k = 0$, there is no attack and $CR(l,0) = R(l)$.
More generally, when an attack is performed, each concealing node receives a reward equal to $R(l+k)$ (since the bogus route has length $l+k$), but the attacker collects both its own reward and the rewards of every Sybil node, for a total reward of $CR(l,k) = (k+1)\cdot R(l+k).$

Analogously, denote by $AC(l,k)$ the net cost obtained by an attacking applicant node, when adding $k \geq 0$ Sybil nodes to $\mathcal{R}$.
If $k = 0$, there is no attack and $AC(l,0) = C(l)$.
More generally, when an attack is performed, the applicant node is charged a total cost of $C(l+k)$ due to the bogus route, but cashes back the share of this cost corresponding to the Sybil nodes, thus effectively paying a total of $AC(l,k) = C(l+k) - k\cdot R(l+k) = l\cdot R(l+k).$

To prevent both kinds of attack simultaneously, the reward scheme must assure that the most beneficial scenario for an attacker occurs when there is no attack, i.e. $k=0$.
Under the scheme assumptions, this amounts to ensuring that the concealer reward is maximized when $k=0$:
$$CR(l,k) \leq CR(l,0)~\forall~ l,k\in\mathbb{N},$$
and the applicant cost is minimized when $k=0$:
$$AC(l,k) \geq AC(l,0)~\forall~ l,k\in\mathbb{N}.$$
Note that, for any $k > 0$, the first condition implies that $R(l+k) \leq R(l)/k+1$, while the second implies that $R(l+k) \geq l\cdot R(l)/l = R(l)$. Both inequalities can only be simultaneously satisfied if $R(l)/(k+1) \geq R(l)$, which is absurd.
\end{proof}

Note that this impossibility results holds for a reward scheme satisfying rather loose conditions.
From the perspective of this analysis, attempts at effectively preventing edge insertion attacks fall in two  categories.


\textit{\textbf{New design relaxing constraints.}}  One option is to design  a reward scheme that relaxes one or more constraints required by Theorem \ref{nao-previne}. In using such approach, one must be mindful of the consequences of the chosen design. For instance, the zero-sum condition can be relaxed by allowing credits to be added or removed from the system depending on the length of the concealment route, but the management of the total credit in the system would likely require some kind of centralized control to prevent the economic collapse of the system, which is certainly not desired in blockchain systems. The length-dependency condition can be relaxed by allowing the reward function to vary over time, but in addition to the possible necessity of exerting this variation in a centralized fashion, the effective result of this strategy would only be to prevent edge insertion attack of only one kind or the other at any given  time instant.

\textit{\textbf{New design with additional security measures.}}  One alternative is to combine the reward scheme with additional security measures. 
In that case, the mechanism may prevent edge insertion attacks by the concealer nodes via a reward scheme,  requiring the last concealer node to validate the concealment route in order to prevent edge insertion attacks by the applicant node.

  {
\emph{\textbf{Discussion. }} 
Although the mechanics of how  fees are charged and paid to mixer nodes are out of the scope of this work, a fundamental assumption in our model is that both kinds of sybil attacks are feasible.  
Route insertion attacks, i.e., adding nodes to $\mathcal{R}$, may be difficult in certain practical settings if the token owner needs to choose a route a priori and specify how much it is willing to pay in fees to the mixer. In   onion routing, for instance, one chooses a route ahead of time, using source routing, and any deviations would result in a router not being able to decrypt the message~\cite{fanti2018dandelion++, tor}. }     

\begin{small}

\bibliographystyle{acm}
\bibliography{main}
\end{small}

\end{document}